\documentclass[nocopyrightspace,computermodern]{sigplanconf}
\usepackage{stmaryrd}
\usepackage{amsmath}
\usepackage{amssymb}
\usepackage{comment}
\usepackage{float}
\usepackage{xspace}
\usepackage{mathabx}
\usepackage{colonequals}
\usepackage[noconfig]{ntheorem}
\usepackage{tikz}
\usepackage{subfigure}
\usepackage{url}
\newtheorem{theorem}{Theorem}[section]
\newtheorem{lemma}[theorem]{Lemma}

\newtheorem{corollary}[theorem]{Corollary}
\newtheorem{definition}[theorem]{Definition}
\newtheorem{example}[theorem]{Example}
\newtheorem*{proof}{Proof}

\newcommand{\SAT}{\mathit{3SAT}}
\newcommand{\CONS}{\mathit{CONS}}
\newcommand{\CS}{\mathit{CS}}
\newcommand{\ME}{\mathit{ME}}
\newcommand{\DME}{\mathit{DME}}
\newcommand{\words}{W_\Sigma}

\newcommand{\shuffle}{\mathbin{\|}}
\renewcommand{\shuffle}{\mathbin{|\hspace{-0.1em}|}}

\newcommand{\lab}{\mathit{lab}}
\renewcommand{\root}{\mathit{root}}

\newcommand{\child}{\mathit{child}}
\newcommand{\ch}{\mathit{ch}}

\newcommand{\true}{\mathit{true}}
\newcommand{\false}{\mathit{false}}

\newcommand{\Tree}{\ensuremath{\mathit{Tree}}}

\newcommand{\ms}{\mathit{MS}}
\newcommand{\dms}{\mathit{DMS}}
\newcommand{\learner}{\mathit{learner}}
\newcommand{\maxcl}{\mathit{max\_clique\_partition}}
\newcommand{\minfit}{\mathit{min\_fit\_multiplicity}}

\newcommand{\cs}{\mathit{CS}}

\def\qed {{                
   \parfillskip=0pt        
   \widowpenalty=10000     
   \displaywidowpenalty=10000  
   \finalhyphendemerits=0  
   \leavevmode             
   \unskip                 
   \nobreak                
   \hfil                   
   \penalty50              
   \hskip.2em              
   \null                   
   \hfill                  
   $\square$
   \par}}                  

\floatstyle{ruled}
\newfloat{float@algorithm}{htbp}{loa}
\floatname{float@algorithm}{Algorithm}
\newcounter{LineCounter@algorithm} 

\newenvironment{BasicCommands@algorithm}{%
  \newcommand{\TAB}{\makebox[4ex][r]{}}%
  \newcommand{\ALGORITHM}{\textbf{algorithm}\xspace}%
  \newcommand{\INPUT}{\textbf{Input}\xspace}%
  \newcommand{\OUTPUT}{\textbf{Output}\xspace}%
  \newcommand{\LET}{\textbf{let}\xspace}%
  \newcommand{\IF}{\textbf{if}\xspace}%
  \newcommand{\THEN}{\textbf{then}\xspace}%
  \newcommand{\FOR}{\textbf{for}\xspace}%
  \newcommand{\DO}{\textbf{do}\xspace}%
  \newcommand{\RETURN}{\textbf{return}\xspace}%
}{%
}

\newenvironment{algorithm*}%
{%
\begin{BasicCommands@algorithm}%
\list{}{\itemindent 0em%
        \listparindent\itemindent
        \rightmargin  \leftmargin}%
\item\relax
}{%
\endlist
\end{BasicCommands@algorithm}%
}

\newenvironment{algorithm}%
{%
  \newcommand{\ResetLineCounter}{%
    \setcounter{LineCounter@algorithm}{0}%
  }%
  \newcommand{\LN}{%
    \makebox[0pt][l]{%
      \makebox[0pt][l]{%
        \addtocounter{LineCounter@algorithm}{1}%
      }%
      \makebox[10.75pt][r]{
        \fontsize{7}{10}%
        \selectfont%
        \arabic{LineCounter@algorithm}:%
      }%
    }%
    \TAB%
  }%
  \begin{BasicCommands@algorithm}%
  \begin{float@algorithm}%
    \ResetLineCounter%
}%
{%
  \end{float@algorithm}%
  \end{BasicCommands@algorithm}%
}

\begin{document}
\title{Learning Schemas for Unordered XML}
\authorinfo{Radu Ciucanu}
           {University of Lille \& INRIA, France}
           {radu.ciucanu@inria.fr}
\authorinfo{S\l awek Staworko}
           {University of Lille \& INRIA, France}
           {slawomir.staworko@inria.fr}
\maketitle

\begin{abstract}
We consider unordered XML, where the relative order among siblings is ignored, and we investigate the problem of learning schemas from examples given by the user.
We focus on the schema formalisms proposed in~\cite{BoCiSt13}: \emph{disjunctive multiplicity schemas} (DMS) and its restriction, \emph{disjunction-free multiplicity schemas} (MS).
A learning algorithm takes as input a set of XML documents
which must satisfy the schema (i.e., \emph{positive examples}) and a set of XML documents which must not satisfy the schema (i.e., \emph{negative examples}), and returns a schema consistent with the examples.
We investigate a learning framework inspired by Gold~\cite{Gold67}, where a learning algorithm should be \emph{sound} i.e., always return a schema consistent with the examples given by the user, and \emph{complete} i.e., able to produce every schema with a sufficiently rich set of examples.
Additionally, the algorithm should be \emph{efficient} i.e., polynomial in the size of the input.
We prove that the DMS are learnable from positive examples only, but they are not learnable when we also allow negative examples.
Moreover, we show that the MS are learnable in the presence of positive examples only, and also in the presence of both positive and negative examples.
Furthermore, for the learnable cases, the proposed learning algorithms return minimal schemas consistent with the examples.
\end{abstract}




\section{Introduction}
\begin{figure}
  \begin{tikzpicture}[yscale=0.875]\small
		\node at (-4,-1) (a) {};
		\node at (0,-1) (n11) {\sl book};
		\node at (-2.5, -2) (m111) {\sl author}edge[-] (n11);
		\node at (-2.5,-2.8) {\scriptsize``$\mathit{C.\,Papadimitriou}$''} edge[-](m111);

		\node at (2.5, -2.1) (m112) {\sl year}edge[-] (n11);
		\node at (2.5,-2.8) {\scriptsize``$\mathit{1994}$''} edge[-] (m112);

		\node at (0, -2) (m113) {\sl title}edge[-] (n11);
		\node at (0,-2.8) {\scriptsize``$\mathit{Computational}$} edge[-] (m113);
		\node at (0,-3.2) {\scriptsize $\mathit{complexity}$''};
  \end{tikzpicture}
\\[0.35cm]
  \begin{tikzpicture}[yscale=0.875]\small
		\node at (-4,-1) (a) {};
		\node at (0,-1) (n10) {\sl book};
		\node at (2.5, -2) (m111) {\sl author}edge[-] (n10);
		\node at (0, -2) (m115) {\sl author}edge[-] (n10);
		\node at (-2.5, -2) (m112) {\sl title}edge[-] (n10);
		\node at (2.5,-2.8) {\scriptsize ``$\mathit{U.\,Vazirani}$''} edge[-](m111);
		\node at (0,-2.8) {\scriptsize ``$\mathit{M.\,Kearns}$''} edge[-](m115);
		\node at (-2.5,-2.8) {\scriptsize ``$\mathit{Computational}$} edge[-] (m112);
		\node at (-2.5,-3.2) {\scriptsize $\mathit{learning\,theory}$''};
  \end{tikzpicture}
\\[0.35cm]
	\begin{tikzpicture}[yscale=0.875,xscale=1.4]\small
		\node at (0,-1) (r) {\sl book};
		\node at (1.2, -2) (n1) {\sl editor} edge[-] (r);
		\node at (1.2, -2.8)  {\scriptsize ``$\mathit{A.\,Bonifati}$''} edge[-] (n1);
		\node at (2.5, -2) (n2) {\sl editor} edge[-] (r);
		\node at (2.5, -2.8)  {\scriptsize ``$\mathit{Z.\,Bellahsene}$''} edge[-] (n2);
		\node at (0.1, -2) (n3) {\sl editor} edge[-] (r);
		\node at (0.1, -2.8)  {\scriptsize ``$\mathit{E.\,Rahm}$''} edge[-] (n3);
		\node at (-1.3, -2) (n4) {\sl title} edge[-] (r);
		\node at (-1.3, -2.8)  {\scriptsize ``$\mathit{Schema\,matching}$} edge[-] (n4);
		\node at (-1.3, -3.2)  {\scriptsize $\mathit{and\,mapping}$''};
		\node at (-2.5, -2.1) (n5) {\sl year} edge[-] (r);
		\node at (-2.5, -2.8)  {\scriptsize ``$\mathit{2011}$''} edge[-] (n5);
	\end{tikzpicture}

  \caption{\label{fig:dblp}Three XML documents storing information about books.}
\end{figure}
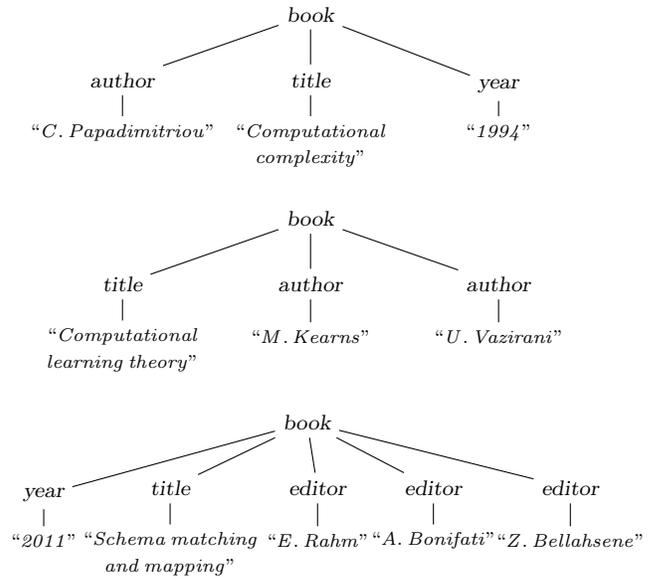
When XML is used for \emph{document-centric} applications, the
relative order among the elements is typically important e.g., the
relative order of paragraphs and chapters in a book. On the other
hand, in case of \emph{data-centric} XML applications, the order among
the elements may be unimportant~\cite{AbBoVi12}. In this paper we
focus on the latter use case. As an example, take in Figure~\ref{fig:dblp} three XML documents storing information about books.
While the order of the elements {\sl title}, {\sl year}, {\sl author}, and {\sl editor} may differ from one {\sl book} to
another, it has no impact on the semantics of the data stored in this
semi-structured database.

A \emph{schema} for XML is a description of the type of admissible
documents, typically defining for every node its content model
i.e., the children nodes it must, may, or cannot contain. 
In this paper we study the problem of \emph{learning} unordered schemas from document examples given by the user.
For instance, consider the three XML documents from Figure~\ref{fig:dblp} and assume that the user wants to obtain a schema which is satisfied by all the three documents.
A desirable solution is a schema which allows a {\sl book} to have, in any order, exactly one {\sl title}, optionally one {\sl year}, and either at least one {\sl author} or at least one {\sl editor}.

Studying the theoretical foundations of learning unordered schemas has several practical motivations.
A schema serves as a reference for users who do not know yet the structure of the XML document, and attempt to query or modify its contents.
If the schema is not given explicitly, it can be learned from document examples and then read by the users.
From another point of view, Florescu~\cite{Florescu05} pointed out the need to automatically infer good-quality schemas and to apply them in the process of \emph{data integration}.
This is clearly a data-centric application, therefore unordered schemas might be more appropriate.
Another motivation of learning the unordered schema of a XML collection is \emph{query minimization}~\cite{ACLS02} i.e., given a query and a schema, find a smaller yet equivalent query in the presence of the schema.
Furthermore, we want to use inferred unordered schemas and optimization techniques to boost the learning algorithms for twig queries~\cite{StWi12},
which are order-oblivious.

Previously, schema learning has been studied from \emph{positive examples} only i.e., documents which must satisfy the schema.
For instance, we have already shown a schema learned from the three documents from Figure~\ref{fig:dblp} given as positive examples.
However, it is conceivable to find applications where \emph{negative examples} (i.e., documents that must not satisfy the schema) might be useful.
For instance, assume a scenario where the schema of a data-centric XML collection evolves over time and some documents may become obsolete w.r.t.\ the new schema. 
A user can employ these documents as negative examples to extract the new schema of the collection.
Thus, the \emph{schema maintenance}~\cite{Florescu05} can be done incrementally, with little feedback needed from the user.
This kind of application motivates us to investigate the problem of learning unordered schemas when we also allow negative examples.

We focus our research on learning the unordered schema formalisms recently proposed in~\cite{BoCiSt13}: the \emph{disjunctive multiplicity schemas} (DMS) and its restriction, \emph{disjunction-free multiplicity schemas} (MS).
While they employ a
user-friendly syntax inspired by DTDs, they define unordered content
model only, and, therefore, they are better suited for unordered XML. 
They also retain much of the expressiveness of DTDs   without an increase in computational complexity.
Essentially, a DMS is a
set of rules associating with each label the possible number of
occurrences for all the allowed children labels by using
\emph{multiplicities}: ``$*$'' (0 or more occurrences), ``$+$'' (1 or
more), ``$?$'' (0 or 1), ``$1$'' (exactly one occurrence; often omitted
for brevity).  Additionally, alternatives can be specified using
restricted \emph{disjunction} (``$\mid$'') and all the conditions are
gathered with \emph{unordered concatenation} (``$\shuffle$''). 
For example, the following schema is satisfied by the three documents from Figure~\ref{fig:dblp}.
\[
{\sl book} \rightarrow  {\sl title} \shuffle {\sl year}^? \shuffle ({\sl author}^+ \mid {\sl editor}^+).
\]
This DMS allows a {\sl book} to have, in any order, exactly one {\sl title}, optionally one {\sl year}, and either at least one {\sl author} or at least one {\sl editor}.
Moreover, this is a \emph{minimal} schema satisfied by the documents from Figure~\ref{fig:dblp} because it captures the most specific schema satisfied by them.
On the other hand, the following schema is also satisfied by the documents from Figure~\ref{fig:dblp}, but it is more general:
\[
{\sl book} \rightarrow {\sl title} \shuffle {\sl year}^? \shuffle {\sl author}^* \shuffle {\sl editor}^*.
\]
This schema allows a {\sl book} to have, in any order, exactly one {\sl title}, optionally one {\sl year}, and any number of {\sl author}'s and {\sl editor}'s.
It is not minimal because it accepts a {\sl book} having at the same time {\sl author}'s and {\sl editor}'s, unlike the first example of schema.
Moreover, the second schema is a MS because it does not use the disjunction operation.

In this paper we address the problem of learning DMS and MS from examples given by the user.
We propose a definition of the learnability influenced by computational learning theory~\cite{KeVa94}, in particular by the inference of languages~\cite{Gold67,Higuera97}.
A learning algorithm takes as input a set of XML documents which must satisfy the schema (i.e., \emph{positive examples}), and a set of XML documents which must not satisfy the schema (i.e., \emph{negative examples}).
Essentially, a class of schemas is \emph{learnable} if there exists an algorithm which takes as input a set of examples given by the user and returns a schema which is consistent with the examples.
Moreover, the learning algorithm should be \emph{sound} i.e., always return a schema consistent with the examples given by the user, \emph{complete} i.e., able to produce every schema with a sufficiently rich set of examples, and \emph{efficient} i.e., polynomial in the size of the input.
Our approach is novel in two directions:
\begin{itemize}
\item Previous research on schema learning has been done in the context of ordered XML, typically on learning restricted classes of regular expressions as content models of the DTDs.
We focus on learning unordered schema formalisms and the results are positive: the DMS and the MS are learnable from positive examples only.
\item The learning frameworks investigated before in the literature typically infer a schema using a collection of documents serving as positive examples.
We study the impact of negative examples in the process of schema learning.
In this case, the learning algorithm should return a schema satisfied by all the positive examples and by none of the negative ones.
We show that the MS are learnable in the presence of both positive and negative examples, while the DMS are not.
\end{itemize}
We summarize our learnability results in Table~\ref{tab:learnability}.
For the learnable cases, we propose learning algorithms which return a minimal schema consistent with the examples.
\begin{table}[tbh]
\begin{center}
\begin{small}
\begin{tabular}{|c|c|c|}
\hline
\emph{Schema formalism} & \emph{+ examples only} & \emph{+ and - examples} \\\hline
DMS & {\bf Yes}~(Th.~\ref{th:dms:pos}) & No~(Th.~\ref{th:dms:pos-neg})\\\hline
MS & {\bf Yes}~(Th.~\ref{th:ms:pos}) & {\bf Yes}~(Th.~\ref{th:ms:posneg})\\\hline
\end{tabular}
\end{small}
\end{center}
\vspace{-1pt}
\caption{\label{tab:learnability}Summary of learnability results.}
\end{table}

\noindent{\bf Related work.}
The \emph{Document Type Definition (DTD)}, the most widespread XML schema formalism \cite{GrMa11, BeNeVa04}, is essentially a set of rules associating with each label a regular expression that defines the admissible sequences of children.
Therefore, learning DTDs reduces to learning regular expressions.
Gold~\cite{Gold67} showed that the entire class of regular languages is not identifiable in the limit.
Consequently, research has been done on restricted classes of regular expressions which can be efficiently learnable~\cite{MiAhCh03}.
Hegewald et al.~\cite{HeNaWe06} extended the approach from~\cite{MiAhCh03} and proposed a system which infers one-unambiguous regular expressions~\cite{BuWo98} as the content models of the labels.
Garofalakis et al.~\cite{GGRSS03} designed a practical system which infers concise and semantically meaningful DTDs from document examples.
Bex et al.~\cite{BNST06, BNSV10} proposed learning algorithms for two classes of regular expressions which capture many practical DTDs and are succinct by definition: \emph{single occurrence regular expressions} (SOREs) and its subclass consisting of \emph{chain regular expressions} (CHAREs).
Bex et al.~\cite{BGNV10} also studied learning algorithms for the subclass of deterministic regular expressions in which each alphabet symbol occurs at most $k$ times ($k$-OREs).
More recently, Freydenberger and K{\"o}tzing~\cite{FrKo13} proposed more efficient algorithms for the above mentioned restricted classes of regular expressions.

Since the DMS disallow repetitions of symbols among the disjunctions, they can be seen as restricted SOREs interpreted under commutative closure i.e., an unordered collection of children matches a regular expression if there exists an ordering that matches the regular expression in the standard way.
The algorithms proposed for the inference of SOREs~\cite{BNSV10,FrKo13} are typically based on constructing an automaton and then transforming it into an equivalent SORE.
Being based on automata techniques, the algorithms for learning SOREs take ordered input, therefore an additional input that the DMS do not have i.e., the order among the labels.
For this reason, we cannot reduce learning DMS to learning SOREs.
Consequently, we have to investigate new techniques to solve the problem of learning unordered schemas.
Moreover, all the existing learning algorithms take into account only positive examples.

We also mention some of the related work on learning schema formalisms more expressive than DTDs. 
\emph{XML Schema}, the second most widespread schema formalism~\cite{GrMa11,BeNeVa04}, allow the content model of an element to depend on the context in which it is used, therefore it is more difficult to learn.
Bex et al.~\cite{BeNeVa07} proposed efficient algorithms to automatically infer a concise XML Schema describing a given set of XML documents.
In a different approach, Chidlovskii~\cite{Chidlovskii01} used \emph{extended context-free grammars} to model schemas for XML and proposed a schema extraction algorithm.\\

\noindent{\bf Organization.} This paper is organized as follows.
In Section~\ref{sec:preliminaries} we present preliminary notions.
In Section~\ref{sec:framework} we formally define the learning framework.
In Section~\ref{sec:dms:pos} and Section~\ref{sec:ms} we present the learnability results for DMS and MS, respectively, when only positive examples are allowed.
In Section~\ref{sec:dms:pos:neg} we discuss the impact of negative examples on learning.
Finally, we summarize our results and outline further directions in Section~\ref{sec:conclusions}.

\section{Preliminaries}\label{sec:preliminaries}
Throughout this paper we assume an alphabet $\Sigma$ which
is a finite set of symbols.
We also assume that $\Sigma$ has a total order $<_\Sigma$, that can be tested in constant time.\\

\noindent {\bf Trees.} We model XML documents with unordered labeled
trees.  Formally, a {\em tree} $t$ is a tuple
$(N_t,\root_t,\lab_t,\child_t)$, where $N_t$ is a finite set of nodes,
$\root_t\in N_t$ is a distinguished root node,
$\lab_t:N_t\rightarrow\Sigma$ is a labeling function, and
$\child_t\subseteq N_t\times N_t$ is the parent-child relation.  We
assume that the relation $\child_t$ is acyclic and require every
non-root node to have exactly one predecessor in this relation.  By
$\Tree$ we denote the set of all finite trees.
We present an example of tree in Figure~\ref{fig:tree}.
\begin{figure}[htb]
    \centering
    \begin{tikzpicture}[yscale=0.85]
      \path[use as bounding box] (-1.25,.25) rectangle (1.25,-3.25);
      \node at (0,0) (n0) {$r$};
      \node at (-0.75,-1) (n1) {$a$};
      \node at (0,-1) (n4) {$b$};
      \node at (0,-2) (n5) {$a$};
      \node at (0.75,-1) (n2) {$c$};
      \node at (0.75,-2) (n3) {$b$};
      \node at (0.75,-3) (n6) {$a$};
      \node at (-0.75,-2) (n7) {$b$};
      \draw[-,semithick] (n1) -- (n7);      
      \draw[-,semithick] (n0) -- (n1);
      \draw[-,semithick] (n0) -- (n1);
      \draw[-,semithick] (n0) -- (n2);
      \draw[-,semithick] (n2) -- (n3);
      \draw[-,semithick] (n3) -- (n6);
      \draw[-,semithick] (n0) -- (n4);
      \draw[-,semithick] (n4) -- (n5);
    \end{tikzpicture}
  \caption{An example of tree.}
  \label{fig:tree}
\end{figure}
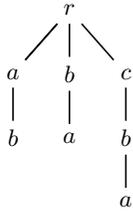

\noindent {\bf Unordered words.} An \emph{unordered word} is
essentially a multiset of symbols i.e., a function
$w:\Sigma\rightarrow\mathbb N_0$ mapping symbols from the alphabet to
natural numbers, and we call $w(a)$ the number of
occurrences of the symbol $a$ in $w$. 
We denote by $\words$ the set containing all the unordered words over the alphabet $\Sigma$.
We also write $a\in w$ as a
shorthand for $w(a)\neq 0$. An empty word $\varepsilon$ is an
unordered word that has $0$ occurrences of every symbol i.e.,
$\varepsilon(a)=0$ for every $a\in\Sigma$. We often use a simple
representation of unordered words, writing each symbol in the alphabet
the number of times it occurs in the unordered word. For example, when
the alphabet is $\Sigma=\{a,b,c\}$, $w_0=aaacc$ stands for the
function $w_0(a) = 3$, $w_0(b) = 0$, and $w_0(c) = 2$. 

The (unordered) concatenation of two unordered words $w_1$ and $w_2$
is defined as the multiset union $w_1\uplus w_2$ i.e., the function
defined as $(w_1\uplus w_2)(a) = w_1(a)+w_2(a)$ for all
$a\in\Sigma$. For instance, $aaacc\uplus{}abbc=aaaabbccc$. Note that
$\varepsilon$ is the identity element of the unordered concatenation
$\varepsilon\uplus w = w\uplus \varepsilon = w$ for all unordered word
$w$. Also, given an unordered word $w$, by $w^i$ we denote the
concatenation $w\uplus\ldots\uplus w$ ($i$ times).

A \emph{language} is a set of unordered words. The unordered
concatenation of two languages $L_1$ and $L_2$ is a language
$L_1\uplus L_2 = \{w_1\uplus w_2\mid w_1\in L_1,\ w_2\in L_2\}$. For
instance, if $L_1 = \{a, aac\}$ and $L_2 = \{ac, b, \varepsilon\}$,
then $L_1\uplus L_2 = \{a,ab,aac,aabc,aaacc\}$.\\

\noindent {\bf Multiplicity schemas.} A \emph{multiplicity} is an element from the set $\{*,+,?,0,1\}$.
\noindent We define the function $\llbracket\cdot\rrbracket$ mapping
multiplicities to sets of natural numbers.  More precisely:
\[
\begin{tabular}{ccccc}
$\llbracket*\rrbracket = \{0,1,2,\ldots\}$, & 
$\llbracket+\rrbracket = \{1,2,\ldots\}$, & 
$\llbracket?\rrbracket = \{0,1\}$, \\
$\llbracket1\rrbracket = \{1\}$, &
$\llbracket0\rrbracket = \{0\}$. 
\end{tabular}
\]
Given a symbol $a\in\Sigma$ and a multiplicity $M$, the language of
$a^M$, denoted $ L(a^M)$, is $\{a^i\mid i\in\llbracket
M\rrbracket\}$.  For example, $ L(a^+) = \{a,aa,\ldots\}$,
$ L(b^0) = \{\varepsilon\}$, and $
L(c^?)=\{\varepsilon,c\}$.

A \emph{disjunctive multiplicity expression} $E$ is:
\[
E \colonequals D_1^{M_1} \shuffle \ldots \shuffle D_n^{M_n},
\]
where for all $1\leq i \leq n$, $M_i$ is a multiplicity and each $D_i$ is:
\[
D_i \colonequals a_1^{M_1'} \mid \ldots \mid a_k^{M_k'},
\]
where for all $1\leq j \leq k$, $M_j'$ is a multiplicity and
$a_j\in\Sigma$.  Moreover, we require that every symbol $a\in \Sigma$
is present at most once in a disjunctive multiplicity expression.  For
instance, $(a\mid b)\shuffle(c\mid d)$ is a disjunctive multiplicity
expression, but $(a\mid b)\shuffle c\shuffle(a\mid d)$ is not because
$a$ appears twice. A \emph{disjunction-free multiplicity expression}
is an expression which uses no disjunction symbol ``$\mid$'' i.e., an
expression of the form $a_1^{M_1} \shuffle \ldots \shuffle a_k^{M_k}$,
where the $a_i$'s are pairwise distinct
symbols in the alphabet and the $M_i$'s are multiplicities (with $1\leq i \leq k$).
We denote by $\DME$ the set of all the disjunctive multiplicity expressions and by $\ME$ the set of all the disjunction-free multiplicity expressions.

The language of a disjunctive multiplicity expression is:
\begin{gather*}
 L(a_1^{M_1} \mid \ldots \mid a_k^{M_k}) =  L(a_1^{M_1}) \cup \ldots\cup L(a_k^{M_k}),\\
 L(D^M) = \{w_1\uplus\ldots\uplus w_i\mid w_1,\ldots,w_i\in L(D)\wedge i\in\llbracket M\rrbracket\},\\
 L(D_1^{M_1} \shuffle \ldots \shuffle D_n^{M_n}) =  L(D_1^{M_1}) \uplus \ldots\uplus  L(D_n^{M_n}).  
\end{gather*}
If an unordered word $w$ belongs to the language of a disjunctive multiplicity expression $E$, we denote it by $w\models E$, and we say that $w$ \emph{satisfies} $E$.
When a symbol $a$ (resp. a disjunctive multiplicity expression $E$)
has multiplicity $1$, we often write $a$ (resp. $E$) instead of $a^1$
(resp. $E^1$).  Moreover, we omit writing symbols and disjunctive
multiplicity expressions with multiplicity $0$. Take, for instance,
$E_0= a^+ \shuffle (b \mid c) \shuffle d^?$ and note that both the
symbols $b$ and $c$ as well as the disjunction $(b \mid c)$ have an
implicit multiplicity $1$. The language of $E_0$ is: 
\[
 L(E_0) = \{a^ib^jc^kd^\ell \mid 
i,j,k,\ell \in \mathbb{N}_0,\ 
i \geq 1,\ 
j+k = 1,\ 
\ell \leq 1
\}.
\]
Next, we recall the unordered schema formalisms from~\cite{BoCiSt13}:
\begin{definition}
  A \emph{disjunctive multiplicity schema (DMS)} is a tuple
  $S=(\root_S, R_S)$, where $\root_S\in\Sigma$ is a designated root
  label and $R_S$ maps symbols in $\Sigma$ to disjunctive multiplicity
  expressions.  By $\dms$ we denote the set of all disjunctive
  multiplicity schemas.  A \emph{disjunction-free multiplicity schema
    (MS)} $S=(\root_S, R_S)$ is a restriction of the $\dms$, where
  $R_S$ maps symbols in $\Sigma$ to disjunction-free multiplicity
  expressions.  By $\ms$ we denote the set of all disjunction-free
  multiplicity schemas.
\end{definition}
To define satisfiability of a DMS $S$ by a tree $t$ we first
define the unordered word $ch_t^n$ of children of a node $n\in N_t$ i.e., 
\[
\ch_t^n(a) = |\{m\in N_t\mid(n,m)\in\child_t\wedge
\lab_t(m)=a\}|.
\]
Now, a tree $t$ \emph{satisfies} $S$, in symbols $t
\models S$, if $\lab_t(\root_t) = \root_S$ and for any node $n\in
N_t$, $\ch_t^n\in L(R_S(\lab_t(n)))$. By $
 L(S)\subseteq \Tree$ we denote the set of all the trees satisfying
$S$. 

In the sequel, we present a schema $S=(\root_S,R_S)$ as a set of
rules of the form $a\rightarrow R_S(a)$, for any $a\in\Sigma$.  If
$ L(R_S(a)) = \varepsilon$, then we write $a\rightarrow
\epsilon$ or we simply omit writing such a rule.
\begin{example}\normalfont We present schemas $S_1,S_2,S_3,S_4$ illustrating the
  formalisms defined above.  They have the root label $r$ and the
  rules:
\begin{align*}
S_1&:~~~r\rightarrow a\shuffle b^*\shuffle c^?&a&\rightarrow b^?&b& \rightarrow a^?&c&\rightarrow b\\[-4pt]
S_2&:~~~r\rightarrow c\shuffle b\shuffle a&a&\rightarrow b^?&b& \rightarrow a&c&\rightarrow b\\[-4pt]
S_3&:~~~r\rightarrow (a\mid b)^+\shuffle c&a&\rightarrow b^?&b& \rightarrow a^?&c&\rightarrow b\\[-4pt]
S_4&:~~~r\rightarrow (a\mid b\mid c)^*&a&\rightarrow \epsilon&b& \rightarrow a^?&c&\rightarrow b
\end{align*}
$S_1$ and $S_2$ are MS, while $S_3$ and $S_4$ are DMS.
The tree from Figure~\ref{fig:tree} satisfies only $S_1$ and $S_3$.
\qed\end{example}
Note that there exist DMS such that the smallest tree in their language has a size exponential in the size of the alphabet, as we observe in the following example.
\begin{example}\normalfont\label{example:exponential}
We consider for $n> 1$ the alphabet $\Sigma=\{r,a_1,b_1,\ldots,a_n,b_n\}$ and the DMS $S_5$ having the root label $r$ and the following rules:
\begin{flalign*}
&r \rightarrow a_1 \shuffle b_1, \\
&a_i\rightarrow a_{i+1} \shuffle b_{i+1}~~(\textrm{for } 1\leq i<n),\\
&b_i\rightarrow a_{i+1} \shuffle b_{i+1}~~(\textrm{for } 1\leq i<n),\\
&a_n\rightarrow \epsilon,\\
&b_n\rightarrow \epsilon.
\end{flalign*}
We present in Figure~\ref{fig:compression} the 
unique tree satisfying this schema and we observe that its size is exponential in the size of the alphabet.
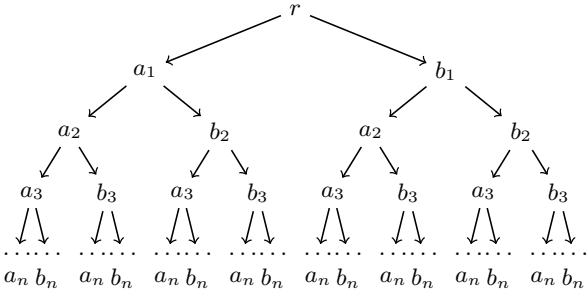
\begin{figure}[h]
	\centering
  \begin{tikzpicture}[yscale=0.65,xscale=2]
    \node at (0,0) (n0) {$r$};
    \node at (-1,-1.25) (n1) {$a_1$};
    \node at (1,-1.25) (n2) {$b_1$};
    \node at (-1.5,-2.5) (n3) {$a_2$};
    \node at (-0.5,-2.5) (n4) {$b_2$};
    \node at (0.5,-2.5) (n30) {$a_2$};
    \node at (1.5,-2.5) (n40) {$b_2$};
    \node at (-1.75,-3.75) (n5) {$a_3$};
    \node at (-1.25,-3.75) (n6) {$b_3$};
    \node at (-0.75,-3.75) (n50) {$a_3$};
    \node at (-0.25,-3.75) (n60) {$b_3$};
    \node at (0.25,-3.75) (n500) {$a_3$};
    \node at (0.75,-3.75) (n600) {$b_3$};
    \node at (1.25,-3.75) (n5000) {$a_3$};
    \node at (1.75,-3.75) (n6000) {$b_3$};
    \node at (-1.85,-5) (n7) {$\ldots$};
    \node at (-1.65,-5) (n71) {$\ldots$};
    \node at (-1.35,-5) (n8) {$\ldots$};
    \node at (-1.15,-5) (n81) {$\ldots$};
    \node at (-0.85,-5) (n70) {$\ldots$};
    \node at (-0.65,-5) (n701) {$\ldots$};
    \node at (-0.35,-5) (n80) {$\ldots$};
    \node at (-0.15,-5) (n801) {$\ldots$};
    \node at (0.15,-5) (n700) {$\ldots$};
    \node at (0.35,-5) (n7001) {$\ldots$};
    \node at (0.65,-5) (n800) {$\ldots$};
    \node at (0.85,-5) (n8001) {$\ldots$};
    \node at (1.15,-5) (n7000) {$\ldots$};
    \node at (1.35,-5) (n70001) {$\ldots$};
    \node at (1.65,-5) (n8000) {$\ldots$};
    \node at (1.85,-5) (n80001) {$\ldots$};

    \node at (-1.85,-5.5) (n7x) {$a_n$};
    \node at (-1.65,-5.5) (n71x) {$b_n$};
    \node at (-1.35,-5.5) (n8x) {$a_n$};
    \node at (-1.15,-5.5) (n81x) {$b_n$};
    \node at (-0.85,-5.5) (n70x) {$a_n$};
    \node at (-0.65,-5.5) (n701x) {$b_n$};
    \node at (-0.35,-5.5) (n80x) {$a_n$};
    \node at (-0.15,-5.5) (n801x) {$b_n$};
    \node at (0.15,-5.5) (n700x) {$a_n$};
    \node at (0.35,-5.5) (n7001x) {$b_n$};
    \node at (0.65,-5.5) (n800x) {$a_n$};
    \node at (0.85,-5.5) (n8001x) {$b_n$};
    \node at (1.15,-5.5) (n7000x) {$a_n$};
    \node at (1.35,-5.5) (n70001x) {$b_n$};
    \node at (1.65,-5.5) (n8000x) {$a_n$};
    \node at (1.85,-5.5) (n80001x) {$b_n$};

    \draw[->,semithick] (n0) -- (n1);
    \draw[->,semithick] (n0) -- (n2);
    \draw[->,semithick] (n1) -- (n3);
    \draw[->,semithick] (n2) -- (n30);
    \draw[->,semithick] (n1) -- (n4);
    \draw[->,semithick] (n2) -- (n40);
    \draw[->,semithick] (n3) -- (n5);
    \draw[->,semithick] (n3) -- (n6);
    \draw[->,semithick] (n4) -- (n50);
    \draw[->,semithick] (n4) -- (n60);
    \draw[->,semithick] (n30) -- (n500);
    \draw[->,semithick] (n30) -- (n600);
    \draw[->,semithick] (n40) -- (n5000);
    \draw[->,semithick] (n40) -- (n6000);
    \draw[->,semithick] (n5) -- (n7);
    \draw[->,semithick] (n5) -- (n71);
    \draw[->,semithick] (n6) -- (n8);
    \draw[->,semithick] (n50) -- (n70);
    \draw[->,semithick] (n60) -- (n80);
    \draw[->,semithick] (n500) -- (n700);
    \draw[->,semithick] (n600) -- (n800);
    \draw[->,semithick] (n5000) -- (n7000);
    \draw[->,semithick] (n6000) -- (n8000);
    \draw[->,semithick] (n6) -- (n81);
    \draw[->,semithick] (n50) -- (n701);
    \draw[->,semithick] (n60) -- (n801);
    \draw[->,semithick] (n500) -- (n7001);
    \draw[->,semithick] (n600) -- (n8001);
    \draw[->,semithick] (n5000) -- (n70001);
    \draw[->,semithick] (n6000) -- (n80001);
  \end{tikzpicture}
\caption{\label{fig:compression}The unique tree satisfying the schema $S_5$.}
\end{figure}
\qed\end{example}

\noindent{\bf Alternative definition with characterizing triples.}
Any disjunctive multiplicity expression $E$ can be expressed alternatively by its \emph{\emph{(}characterizing\emph{)} triple} $(C_E,N_E,P_E)$ consisting of the following sets:
\begin{itemize}
\item The \emph{conflicting pairs of siblings} $C_E$ contains
  pairs of symbols in $\Sigma$ such that $E$ defines no word using
  both symbols simultaneously:
  \[
  C_E = \{(a_1,a_2)\in\Sigma\times\Sigma\mid \varnot \exists
  w\in L(E).\ a_1\in w\wedge a_2\in w\}.
  \]
\item The \emph{extended cardinality map} $N_E$ captures for each symbol
 in the alphabet the possible numbers of its occurrences in
  the unordered words defined by $E$:
  \[
  N_{E} = \{(a, w(a))\in \Sigma\times\mathbb N_0\mid w\in  L(E)\}.
  \]
\item The \emph{sets of required symbols} $P_E$ which captures symbols
  that must be present in every word; essentially, a set of symbols $X$
  belongs to $P_E$ if every word defined by $E$ contains at least one
  element from $X$:
  \[
  P_E=\{X\subseteq\Sigma\mid\forall w\in L(E).\ \exists a\in
  X.\ a \in w\}.
  \]
\end{itemize}
As an example we take $E_0= a^+ \shuffle (b \mid c) \shuffle d^?$.
Because $P_E$ is closed under supersets, we list only its minimal
elements:
\begin{gather*}
  C_{E_0} = \{ (b,c), (c,b) \}, \qquad 
  P_{E_0} = \{ \{a\}, \{b,c\}, \ldots\},\\
  N_{E_0} = \{ (b,0), (b,1), (c,0), (c,1), (d,0), (d,1), (a,1), (a,2), 
  \ldots\}.
\end{gather*}
Two equivalent disjunctive multiplicity expressions yield the same triples and hence $(C_E,N_E,P_E)$ can be viewed as the \emph{normal form} of a given expression $E$~\cite{BoCiSt13}.
Moreover, each set has a compact representation of size polynomial in the size of the alphabet and computable in PTIME. 
We illustrate them on the same $E_0= a^+ \shuffle (b \mid c) \shuffle d^?$:
\begin{itemize}
\item $C_E^*$ consists of sets of symbols present in $E$ such that any pairwise two of them are conflicting:
\[
C_{E_0}^* = \{\{b,c\}\}.
\]
\item $N_E^*$ is a function mapping symbols to multiplicities such that for any unordered word $w\in L(E)$, and for any symbol $a\in\Sigma$, $w(a)\in \llbracket N_E^*(a)\rrbracket$:
\[
N_{E_0}^*(a) = +,~~~N_{E_0}^*(b) = N_{E_0}^*(c) = N_{E_0}^*(d) = {?}.
\]
\item $P_E^*$ contains only the $\subseteq$-minimal elements of $P_E$:
\[
P_{E_0}^* = \{\{a\}, \{b,c\}\}.
\]
\end{itemize}
Also note that we can easily construct a disjunctive multiplicity expression from its characterizing triple.
A simple algorithm has to loop over the sets from $C_E^*$ and $P_E^*$ to compute for each label with which other labels it is linked by the disjunction operator.
Then, using $N_E^*$, the algorithm associates to each label and each disjunction the correct multiplicity.
For example, take the following compact triples:
\begin{gather*}
C_{E_1}^* = \{\{a,e\},\{c,d\}\}, \qquad P_{E_1}^* = \{\{a,e\},\{b\}\},\\
N_{E_1}^*(a) = *,~~ N_{E_1}^*(b) = 1,~~N_{E_1}^*(c)=N_{E_1}^*(d)=N_{E_1}^*(e) = {?}.
\end{gather*}
Note that they characterize the expression:
\[
{E_1}=(a^+\mid e)\shuffle b\shuffle (c^?\mid d^?).
\]
We have introduced the alternative definition with characterizing triples because we later propose an algorithm which learns characterizing triples from unordered word examples (Algorithm~\ref{alg1} from Section~\ref{sec:dms:pos}).
Then, from this information, the corresponding disjunctive multiplicity expression can be constructed in a straightforward manner.

\section{Learning framework}\label{sec:framework}
We use a variant of the standard language inference framework~\cite{Gold67, Higuera97} adapted to learning disjunctive multiplicity expressions and schemas.
A \emph{learning setting} is a tuple containing the set of \emph{concepts} that are to be learned, the set of \emph{instances} of the concepts that are to serve as examples in learning, and the \emph{semantics} mapping every concept to its set of instances.
\begin{definition}
A \emph{learning setting} is a tuple $(\mathcal{E,C,L})$, where $\mathcal E$ is a set of examples, $\mathcal C$ is a class of concepts, and $\mathcal L$ is a function that maps every concept in $\mathcal C$ to the set of all its examples (a subset of $\mathcal E$).
\end{definition}
For example, the setting for learning disjunctive multiplicity expressions from positive examples is the tuple $(\words,\DME, L)$ and the setting for learning disjunctive multiplicity schemas from positive examples is $(\Tree,\dms, L)$.
We obtain analogously the learning settings for disjunction-free multiplicity expressions and schemas: $(\words,\ME,L)$ and $(\Tree, \ms,L)$, respectively.
The general formulation of the definition allows us to easily define settings for learning from both positive and negative examples, which we present in Section~\ref{sec:dms:pos:neg}.

To define a learnable concept, we fix a learning setting $\mathcal K = (\mathcal E,\mathcal C,\mathcal L)$ and we introduce some auxiliary notions.
A \emph{sample} is a finite nonempty subset $D$ of $\mathcal E$ i.e., a set of examples.
A sample $D$ is \emph{consistent} with a concept $c\in\mathcal C$ if $D\subseteq\mathcal L(c)$.
A \emph{learning algorithm} is an algorithm that takes a sample and returns a concept in $\mathcal C$ or a special value \emph{null}.
\begin{definition}\label{def:learnable}
A class of concepts $\mathcal C$ is \emph{learnable in polynomial time and data} in the setting $\mathcal K = (\mathcal{E,C,L})$ if there exists a polynomial learning algorithm $\learner$ satisfying the following two conditions:
\begin{enumerate}
\item {\bf Soundness.} For any sample $D$, the algorithm $\learner(D)$ returns a concept consistent with $D$ or a special \emph{null} value if no such concept exists.
\item {\bf Completeness.} For any concept $c\in\mathcal C$ there exists a sample $\CS_c$ such that for every sample $D$ that extends $\CS_c$ consistently with $c$ i.e., $\CS_c\subseteq D\subseteq \mathcal L(c)$, the algorithm $\learner(D)$ returns a concept equivalent to $c$.
Furthermore, the cardinality of $\CS_c$ is polynomially bounded by the size of the concept. 
\end{enumerate}
\end{definition}
The sample $\cs_c$ is called the \emph{characteristic sample} for $c$ w.r.t.\ $\learner$ and $\mathcal K$.
For a learning algorithm there may exist many such samples.
The definition requires that one characteristic sample exists.
The soundness condition is a natural requirement, but alone it is not sufficient to eliminate trivial learning algorithms.
For instance, if we want to learn disjunctive multiplicity expressions from positive examples over the alphabet $\{a_1,\ldots,a_n\}$, an algorithm always returning $a_1^*\shuffle\ldots\shuffle a_n^*$ is sound.
Consequently, we require the algorithm to be complete analogously to how it is done for grammatical language inference~\cite{Gold67, Higuera97}.

Typically, in the case of polynomial grammatical inference, the \emph{size} of the characteristic sample is required to be polynomial in the size of the concept to be learned~\cite{Higuera97}, where the size of a sample is the sum of the sizes of the examples that it contains.
From the definition of the DMS, since repetitions of symbols are discarded among the disjunctions, the size of a schema is polynomial in the size of the alphabet.
Thus, a natural requirement would be that the size of the characteristic sample is polynomially bounded by the size of the alphabet.
There exist DMS such that the smallest tree in their language is exponential in the size of the alphabet (cf. Example~\ref{example:exponential}).
Because of space restrictions, we have imposed in the definition of learnability that the \emph{cardinality} (and not the size) of the characteristic sample is polynomially bounded by the size of the concept, hence by the size of the alphabet.
However, we are able to obtain characteristic samples of size polynomial in the size of the alphabet by using a \emph{compressed} representation of the XML trees, for example with \emph{directed acyclic graphs}~\cite{LoMaNo13}.
We will provide in the full version of the paper the details about this compression technique and the new definition of the learnability.
The algorithms that we propose in this paper transfer without any alteration for the definition using compressed trees.

Additionally to the conditions imposed by the definition of learnability, we are interested in the existence of learning algorithms which return \emph{minimal} concepts for a given set of examples.
It is important to emphasize that we mean minimality in terms on language inclusion.
When only positive examples are allowed, a DMS $S$ is a \emph{minimal} DMS consistent with a set of trees $D$ iff $D\subseteq L(S)$, and, for any $S'\neq S$, if $D\subseteq L(S')$, then $L(S')\varnot\subset L(S)$.
We similarly obtain the definition of minimality for learning disjunctive multiplicity expressions.
Intuitively, a minimal schema consistent with a set of examples is the most specific schema consistent with them.
For example, recall the three XML documents storing information about books from Figure~\ref{fig:dblp}.
Assume that the user provides the three documents as positive examples to a learning algorithm.
The most specific schema consistent with the examples is:
\[
{\sl book} \rightarrow  {\sl title} \shuffle {\sl year}^? \shuffle ({\sl author}^+ \mid {\sl editor}^+).\]
Another possible solution is the schema:
\[
{\sl book} \rightarrow  {\sl title} \shuffle {\sl year}^? \shuffle {\sl author}^* \shuffle {\sl editor}^*.
\]
It is less likely that a user wants to obtain such a schema which allows a {\sl book} to have at the same time {\sl author}'s and {\sl editor}'s.
In this case, the most specific schema also corresponds to the natural requirements that one might want to impose on a XML collection storing information about books, in particular a {\sl book} has either at least one {\sl author} or at least one {\sl editor}.
Minimality is often perceived as a better fitted learning solution~\cite{Angluin80,Angluin82,BGNV10,GaVi90}, and this motivates our requirement for the learning algorithms to return minimal concepts consistent with the examples.

\section{Learning DMS from positive examples}\label{sec:dms:pos}
The main result of this section is the learnability of the disjunctive multiplicity schemas from positive examples i.e., in the setting $(\Tree,\dms, L)$.
We present a learning algorithm that constructs a \emph{minimal} schema consistent with the input set of trees.

First, we study the problem of learning a disjunctive multiplicity expression from positive examples i.e., in the setting $(W_\Sigma, \DME,L)$.
We present a learning algorithm that constructs a minimal disjunctive multiplicity expression consistent with the input collection of unordered words.
Given a set of unordered words, there may exist many consistent minimal disjunctive multiplicity expressions.
In fact, for some sets of positive examples there may be an exponential number of such expressions (cf. the proof of Lemma~\ref{lemma:cons}).
Take in Example~\ref{ex:minimal} a sample and two consistent minimal disjunctive multiplicity expressions.

\begin{example}\label{ex:minimal}\normalfont
Consider the alphabet $\Sigma=\{a,b,c,d,e\}$ and the set of unordered words $D=\{aabc, abd, be\}$.
Take the following two disjunctive multiplicity expressions:
\begin{gather*}
E_1 = (a^+\mid e)\shuffle b \shuffle (c^?\mid d^?),\\
E_2 = a^*\shuffle b  \shuffle (c\mid d\mid e).
\end{gather*}
Note that $D\subseteq  L(E_1)$ and $D\subseteq  L(E_2)$.
Also note that $ L(E_1)\varnot\subseteq  L(E_2)$ (because of $bce$) and $ L(E_2)\varnot\subseteq  L(E_1)$ (because of $abe$).
On the other hand, we easily observe that both $E_1$ and $E_2$ are minimal disjunctive multiplicity expressions with languages including $D$.\qed
\end{example}
Before we present the learning algorithms, we have to introduce additional notions.
First, we define the function $\minfit(\cdot)$ which, given a set of unordered words $D$ and a label $a\in\Sigma$, computes the multiplicity $M$ such that $\forall w\in D.\ w(a)\in\llbracket M\rrbracket$ and there does not exist another multiplicity $M'$ such that $\llbracket M'\rrbracket\subset\llbracket M\rrbracket$ and $\forall w\in D.\ w(a)\in \llbracket M'\rrbracket$.
For example, given the set of unordered words $D=\{aabc, abd, be\}$, we have:
\begin{gather*}
\minfit(D,a) = *, \\
\minfit(D,b) = 1,\\
\minfit(D,c) = {?}.
\end{gather*}
Next, we introduce the notion of \emph{maximal-clique partition of a graph}.
Given a graph $G=(V,E)$, a maximal-clique partition of $G$ is a graph partition $(V_1,\ldots,V_k)$ such that:
\begin{itemize} 
\item The subgraph induced in $G$ by any $V_i$ is a clique (with $1\leq i\leq k$),
\item The subgraph induced in $G$ by the union of any $V_i$ and $V_j$ is not a clique (with $1\leq i\neq j\leq k$).
\end{itemize}
In Figure~\ref{fig:cuts} we present a graph and a maximal-clique partition of it i.e., $\{\{a,e\}, \{b\}, \{c,d\}\}$.
Note that the graph from Figure~\ref{fig:cuts} allows one other maximal-clique partition i.e., $\{\{a\}, \{b\}, \{c,d,e\}\}$.
On the other hand, $\{\{a\}, \{b\}, \{c,d\},\{e\}\}$ is not a maximal-clique partition because it contains two sets such that their union induces a clique i.e., $\{a\}$ and $\{e\}$.
\begin{figure}[htb]
\centering
	\begin{tikzpicture}[xscale=0.85, yscale=0.85]
		\node at (-0.75, 0.75) (a) {$a$};
		\node at (0.75, 0.75) (e) {$e$};
		\node at (-0.75,-0.75) (c) {$c$};
		\node at (0.75,-0.75) (d) {$d$};
		\node at (-2,0) (b) {$b$};
    \draw (a) edge[-] (e);
    \draw (e) edge[-] (c);
    \draw (e) edge[-] (d);
    \draw (c) edge[-] (d);
\draw [rounded corners, densely dotted, red, very thick] (-1.1, 1.1) rectangle (1.1, 0.4);
\draw [rounded corners, densely dotted, red, very thick] (1.1, -1.1) rectangle (-1.1, -0.4);
\draw [rounded corners, densely dotted, red, very thick] (-2.35,-0.35) rectangle (-1.65, 0.35);
	\end{tikzpicture}
	\caption{A graph and a maximal-clique partition of it. Vertices from the same rectangle belong to the same set.\label{fig:cuts}}
\end{figure}
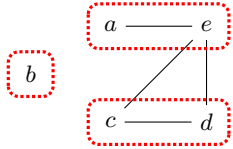

\noindent Unlike the \emph{clique} problem, which is known to be NP-complete~\cite{Papadimitriou94}, we can partition in PTIME a graph in maximal cliques with a greedy algorithm.
In the sequel, we assume that the vertices of the graph are labels from $\Sigma$.
For a given graph there may exist many maximal-clique partitions and we use the total order $<_\Sigma$ to propose a deterministic algorithm constructing a maximal-clique partition.
The algorithm works as follows: we take the smallest label from $\Sigma$ w.r.t.\ $<_\Sigma$ and not yet used in a clique, and we iteratively extend it to a maximal clique by adding connected labels.
Every time when we have a choice to add a new label to the current clique, we take the smallest label w.r.t.\ $<_\Sigma$.
We repeat this until all the labels are used.
This algorithm yields to a unique maximal-clique partition.
For example, for the graph from Figure~\ref{fig:cuts}, we compute the maximal-clique partition marked on the figure i.e., $\{\{a,e\}, \{b\}, \{c,d\}\}$.
We additionally define the function $\maxcl(\cdot)$ which takes as input a graph, computes a maximal-clique partition using the greedy algorithm described above and, at the end, for technical reasons, the algorithm discards the singletons.
For example, for the graph from Figure~\ref{fig:cuts}, the function $\maxcl(\cdot)$ returns $\{\{a,e\},\{c,d\}\}$.
Clearly, the function $\maxcl(\cdot)$ works in PTIME.

Next, we present Algorithm~\ref{alg1} and we claim that, given a set of unordered words $D$, it computes in polynomial time a disjunctive multiplicity expression $E$ consistent with $D$.
\begin{algorithm}
\caption{\label{alg1}Learning disjunctive multiplicity expressions from positive examples.}
\ALGORITHM $\learner_\DME^+(D)$\\
\INPUT: A set of unordered words $D = \{w_1,\ldots,w_n\}$\\
\OUTPUT: A minimal disjunctive multiplicity expression $E$ consistent with $D$\\
\LN \FOR $a\in\Sigma$ \DO\\
\LN \TAB \LET $N_E^*(a) = \minfit(D,a)$\\
\LN \LET $\Sigma' = \{a\in \Sigma\mid N_E^*(a)\in\{?,1,*,+\}\}$\\
\LN \LET $G = (\Sigma',~ \{(a,b)\in\Sigma'\times\Sigma'\mid \forall w\in D.\ a\notin w\vee b\notin w\})$\\
\LN \LET $C_E^* = \maxcl(G)$\\
\LN \LET $P_E^* = \{\{a\}\mid N_E^*(a) \in \{1, +\}\}\\
\TAB\TAB\TAB\cup\{X\in C_E^*\mid\forall w\in D.\ \exists a\in X.\ a\in w\} $\\
\LN \RETURN $E$ characterized by the triple $(C_E^*, N_E^*, P_E^*)$
\end{algorithm}

\noindent Algorithm~\ref{alg1} works in three steps and we illustrate each of them on the sample $D= \{aabc,abd,be\}$ from Example~\ref{ex:minimal}.
The first step (lines 1-2) computes the compact representation of the extended cardinality map for each symbol from $\Sigma$, using the function $\minfit(\cdot)$.
We ignore in the sequel the symbols never occurring in words from $D$ (line 3).
For the sample from Example~\ref{ex:minimal}, we infer:
\begin{gather*}
N_E^*(a) = *, \qquad N_E^*(b) = 1,\\N_E^*(c)=N_E^*(d)=N_E^*(e) = {?}.
\end{gather*}
The second step of the algorithm (lines 4-5) computes the compact sets of conflicting siblings.
First, we construct the graph $G$ having as set of vertices the labels occurring at least once in unordered words from $D$.
Two labels are linked by an edge in $G$ if there does not exist an unordered word in $D$ where both of them are present at the same time, in other words the two labels are a candidate pair of conflicting siblings.
Next, we apply the function $\maxcl(\cdot)$ on the graph $G$.
For the unordered words from Example~\ref{ex:minimal} we obtain the graph from Figure~\ref{fig:cuts}, and we infer $C_E^*= \{\{a,e\}, \{c,d\}\}$.
Note that the maximal-clique partition implies the minimality of the disjunctive multiplicity expression constructed later using the inferred $C_E^*$.

The third step of the algorithm (line 6) computes the $\subseteq$-minimal sets of required symbols $P_E^*$.
Each symbol having associated a multiplicity $1$ or $+$ belongs to a required set of symbols containing only itself because it is present in all the unordered words from $D$ and we want to learn a minimal concept.
Moreover, we add in $P_E^*$ the sets of conflicting siblings inferred at the previous step with the property that one of them is present in any unordered word from $D$, to guarantee the minimality of the inferred language.
For the sample from Example~\ref{ex:minimal},
$\{b\}$ belongs to $P_E^*$.
Since from the previous step we have $C_E^* = \{\{a,e\}, \{c,d\}\}$, at this step we have to add $\{a,e\}$ to $P_E^*$ because all the words in the sample contain either $a$ or $e$.
On the other hand, we do not add $\{c,d\}$ because the sample contains the word $be$.
The inferred $P_E^*$ is $\{\{a,e\}, \{b\}\}$.

Finally, the algorithm returns the disjunctive multiplicity expression characterized by the inferred triple (line 7).
For the sample $D$, it returns $E=(a^+\mid e)\shuffle b \shuffle (c^?\mid d^?)$.
Note that if at step 2 we take a partition which is not a maximal-clique one, for example $\{\{a\}, \{b\}, \{c,d\},\{e\}\}$, and we later construct a disjunctive multiplicity expression using it, we get $a^*\shuffle b\shuffle (c^?\mid d^?)\shuffle e^?$, which includes both $E_1$ and $E_2$ from Example~\ref{ex:minimal}, therefore is not minimal.
Also note that at step 3, without $\{a,e\}$ added to $P_E^*$, the resulting schema would accept an unordered word without any $a$ and $e$, so the learned language would not be minimal.

Algorithm~\ref{alg1} is sound and each of its three steps requires polynomial time.
Next, we prove the completeness of the algorithm.
Given a disjunctive multiplicity expression $E$, we construct in three steps its characteristic sample $\cs_E$.
At the same time, we illustrate the construction on  the disjunctive multiplicity expression $E_1 = (a^+\mid e)\shuffle b \shuffle (c^?\mid d^?)$:
\begin{enumerate}
\item We take the pairs of symbols which can be found together in an unordered word in $L(E)$.
For each of them, we add in $\cs_E$ an unordered word containing only the two symbols.
Next, for each symbol occurring in the disjunctions from $E$, we add in $\cs_E$ an unordered word containing only one occurrence of that symbol.
We also add in $\cs_E$ the empty word. 
For $E_1$ we obtain:
$\{ab,ac,ad, bc,bd, be, ce,de, a,b,c,d,e,\varepsilon\}$.
\item We replace each unordered word $w$ obtained at the previous step with $w\uplus w'$, where $w'$ is a minimal unordered word such that $w\uplus w'\in L(E)$.
The newly obtained $\cs_E$ contains unordered words from $L(E)$.
For $E_1$ we obtain:
$\{ab,abc,abd,be,bce,bde\}$.
\item For each symbol $a$ from the alphabet such that $N_E^*(a)$ is $*$ or $+$, we randomly take an unordered word $w$ from $\cs_E$ and containing $a$ and we add to $\cs_E$ the unordered word $w\uplus a$.
In the worst case, at this step the number of words in the characteristic sample is doubled, but it remains polynomial in the size of the alphabet.
For $E_1$ we obtain:
$\{ab,aab,abc,abd,be,bce,bde\}$.
\end{enumerate}
Note that there may exist many equivalent characteristic samples.
The first step of the construction implies that the only potential conflicts to be considered in Algorithm~\ref{alg1} are the conflicts implied by the expression.
In other words, all the connected components of the graph of potential conflicts from Algorithm~\ref{alg1} are cliques.
Thus, there is only one possible maximal-clique partition to be done in the algorithm.
Moreover, the second and third steps of the construction ensure that, for any sample consistently extending the characteristic sample, Algorithm~\ref{alg1} infers the correct sets of required symbols and the extended cardinality map, respectively.

We have proposed Algorithm~\ref{alg1}, which is a sound and complete algorithm for learning minimal disjunctive multiplicity expressions from unordered words positive examples.
Thus, we can state the following result:
\begin{lemma}
The concept class $\DME$ is learnable in polynomial time and data from positive examples i.e., in the setting $(\words,\DME,L)$.
\end{lemma}
Next, we extend the result for DMS.
We propose Algorithm~\ref{alg2}, which learns a disjunctive multiplicity schema from a set of trees.
We assume w.l.o.g. that all the trees from the sample have as root label the same label $r$.
If this assumption is not satisfied, the sample is not consistent.
The algorithm infers, for each label $a$ from the alphabet, the minimal disjunctive multiplicity expression consistent with the children of all the nodes labeled $a$ from the trees from the sample.
\begin{algorithm}
\caption{\label{alg2}Learning DMS from positive examples.}
\ALGORITHM: $\learner_\dms^+(D)$\\
\INPUT: A set of trees $D = \{t_1,\ldots,t_n\}$ s.t.\ $\lab_{t_i}(\root_{t_i}) = r$ (with $1\leq i\leq n)$\\
\OUTPUT: A minimal DMS $S$ consistent with $D$\\
\LN \FOR $a\in\Sigma$ \DO\\
\LN \TAB \LET $D'=\{\ch_t^n\mid t\in D.\ n\in N_t.\ \lab_t(n) = a\}$\\
\LN \TAB \LET $R_S(a) = \learner_\DME^+(D')$\\
\LN \RETURN $S=(r,R_S)$
\end{algorithm}

\noindent Algorithm~\ref{alg2} returns a minimal disjunctive multiplicity schema consistent with the sample because the inferred rule for each label represents a minimal disjunctive multiplicity expression obtained using Algorithm~\ref{alg1}.
Next, we show that Algorithm~\ref{alg2} is also complete by  providing a construction of a characteristic sample of cardinality polynomial in the size of the alphabet.
For this purpose, we have to define first two additional notions.
Given a DMS $S=(\root_S,R_S)$ and a label $a\in\Sigma$, we define the following two trees:
\begin{itemize}
\item $\min_{t\uparrow(S,a)}$ is a minimal tree satisfying $S$ and containing a node labeled $a$,
\item $\min_{t\downarrow(S,a)}$ is a minimal tree satisfying $S'=(a,R_S)$. 
It is equivalent to $\min_{t\uparrow(S',a)}$.
\end{itemize}
We illustrate the two notions defined above in the following example:
\begin{example}\label{ex:min:tree}\normalfont
Consider the DMS $S$ having the root label $r$ and the rules:
\begin{gather*}
r\rightarrow a^*\shuffle (b\mid c)\qquad a\rightarrow d^?\\
b,c\rightarrow e^+ \qquad~~~~~~~ d,e\rightarrow\epsilon
\end{gather*}
We present in Figure~\ref{fig:min:trees} some trees and we explain for each of them how it can be used.
\begin{figure}[htb]
\centering
\subfigure[\newline$\min_{t\downarrow(S,r)}$\newline$\min_{t\uparrow(S,r)}$\newline$\min_{t\uparrow(S,b)}$\newline$\min_{t\uparrow(S,e)}$]{
	\begin{tikzpicture}[yscale=0.85]
		\node at (0,0) (r) {$r$};
		\node at (0,-1) (b) {$b$};
		\node at (0,-2) (e) {$e$};
		\node at (1,0) (x) {};
		\draw (r) edge[-] (b);
		\draw (b) edge[-] (e);
	\end{tikzpicture}
}
\subfigure[\newline$\min_{t\downarrow(S,r)}$\newline$\min_{t\uparrow(S,r)}$\newline$\min_{t\uparrow(S,c)}$\newline$\min_{t\uparrow(S,e)}$]{
	\begin{tikzpicture}[yscale=0.85]
		\node at (0,0) (r) {$r$};
		\node at (0,-1) (b) {$c$};
		\node at (0,-2) (e) {$e$};
		\node at (1,0) (x) {};
		\draw (r) edge[-] (b);
		\draw (b) edge[-] (e);
	\end{tikzpicture}
}
\subfigure[$\min_{t\uparrow(S,a)}$]{
	\begin{tikzpicture}[yscale=0.85,xscale=0.5]
		\node at (0,0) (r) {$r$};
		\node at (-0.5,-1) (a) {$a$};
		\node at (0.5,-1) (b) {$b$};
		\node at (0.5,-2) (e) {$e$};
		\node at (1.5,0) (x) {};
		\draw (r) edge[-] (b);
		\draw (r) edge[-] (a);
		\draw (b) edge[-] (e);
	\end{tikzpicture}
}
\subfigure[$\min_{t\uparrow(S,d)}$]{
	\begin{tikzpicture}[yscale=0.85,xscale=0.5]
		\node at (0,0) (r) {$r$};
		\node at (-0.5,-1) (a) {$a$};
		\node at (-0.5,-2) (d) {$d$};
		\node at (0.5,-1) (b) {$b$};
		\node at (0.5,-2) (e) {$e$};
		\node at (1.5,0) (x) {};
		\draw (r) edge[-] (b);
		\draw (r) edge[-] (a);
		\draw (b) edge[-] (e);
		\draw (a) edge[-] (d);
	\end{tikzpicture}
}
\subfigure[ $\min_{t\downarrow(S,a)}$]{
	\begin{tikzpicture}[yscale=0.85]
		\node at (0,0) (r) {};
		\node at (0,-1) (b) {$a$};
		\node at (0,-2) (e) {$d$};
		\node at (1,0) (x) {};
		\draw (b) edge[-] (e);
	\end{tikzpicture}
}
\subfigure[ $\min_{t\downarrow(S,b)}$]{
	\begin{tikzpicture}[yscale=0.85]
		\node at (0,0) (r) {};
		\node at (0,-1) (b) {$b$};
		\node at (0,-2) (e) {$e$};
		\node at (1,0) (x) {};
		\draw (b) edge[-] (e);
	\end{tikzpicture}
}
\subfigure[ $\min_{t\downarrow(S,c)}$]{
	\begin{tikzpicture}[yscale=0.85]
		\node at (0,0) (r) {};
		\node at (0,-1) (b) {$c$};
		\node at (0,-2) (e) {$e$};
		\node at (1,0) (x) {};
		\draw (b) edge[-] (e);
	\end{tikzpicture}
}
\subfigure[ $\min_{t\downarrow(S,d)}$]{
	\begin{tikzpicture}[yscale=0.85]
		\node at (0,0) (r) {};
		\node at (0,-1) (b) {};
		\node at (0,-2) (e) {$d$};
		\node at (1,0) (x) {};
	\end{tikzpicture}
}
\subfigure[ $\min_{t\downarrow(S,e)}$]{
	\begin{tikzpicture}[yscale=0.85]
		\node at (0,0) (r) {};
		\node at (0,-1) (b) {};
		\node at (0,-2) (e) {$e$};
		\node at (1,0) (x) {};
	\end{tikzpicture}
}
\caption{Trees used for Example~\ref{ex:min:tree}.\label{fig:min:trees}}
\end{figure}
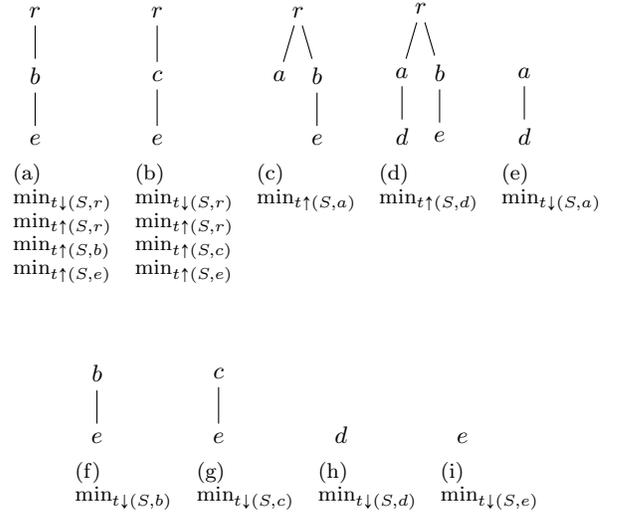
\qed\end{example}
Next, we present the construction of the characteristic sample for learning a DMS from positive examples.
We take a DMS $S=(\root_S,R_S)$ over an alphabet $\Sigma$ and we assume w.l.o.g.\ that any symbol of the alphabet can be present in at least one tree from $L(S)$.
For each $a\in\Sigma$, for each $w\in\cs_{R_S(a)}$, we compute a tree $t$ as follows: we generate a tree $\min_{t\uparrow(S,a)}$, we take the node labeled by $a$ (let it $n_a$), and for any $b\in\Sigma$, while $\ch_t^{n_a}(b)<w(b)$ we fuse in $n_a$ a copy of $\min_{t\downarrow(S,b)}$.
We obtain a sample of cardinality polynomially bounded by the size of the alphabet.
Given a DMS $S$, there may exist many characteristic samples $\cs_S$.
Each of them has the property that, if we construct a sample $D$ which extends $\cs_S$ consistently with $S$, then $\learner_\dms^+(D)$ returns $S$.
This proves the completeness of Algorithm~\ref{alg2}.

We illustrate the construction of the characteristic sample on the schema $S$ from Example~\ref{ex:min:tree}.
Recall that we have already presented the trees $\min_{t\uparrow(S,a)}$ and $\min_{t\downarrow(S,a)}$ for each $a$ from the alphabet.
We also construct the characteristic samples for the disjunctive multiplicity expressions from the rules of $S$:
\begin{itemize}
\item $\cs_{R_S(r)} = \{aab,ab,ac,b,c\}$,
\item $\cs_{R_S(a)} = \{\varepsilon,d\}$,
\item $\cs_{R_S(b)} = \cs_{R_S(c)} = \{e,ee\}$,
\item $\cs_{R_S(d)} = \cs_{R_S(e)} = \{\varepsilon\}$.
\end{itemize}
In Figure~\ref{fig:cs} we present a characteristic sample $\cs_S$ for the DMS $S$ and we explain the purpose of each tree:

\begin{itemize}
\item (a), (b), (c), (d), and (e) ensure that there is inferred the correct rule for the root i.e., $R_S(r)$,
\item (b) and (f) ensure that there is inferred the correct $R_S(a)$,
\item (d) and (g) ensure that there is inferred the correct $R_S(b)$,
\item (e) and (h) ensure that there is inferred the correct $R_S(c)$,
\item The nodes labeled by $d$ and $e$ never have children in the trees from $\cs_S$, so there are inferred the correct rules for $R_S(d)$ and  $R_S(e)$. 
\end{itemize}
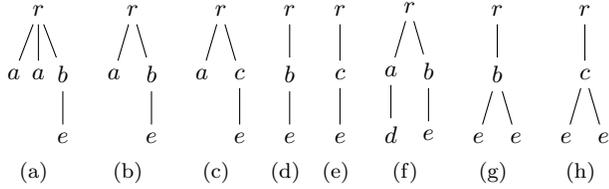
\begin{figure}[h]
\subfigure[]{
	\begin{tikzpicture}[yscale=0.85,xscale=0.65]
		\node at (0,0) (r) {$r$};
		\node at (-0.5,-1) (a) {$a$};
		\node at (0,-1) (a2) {$a$};
		\node at (0.5,-1) (b) {$b$};
		\node at (0.5,-2) (e) {$e$};
		\draw (r) edge[-] (b);
		\draw (r) edge[-] (a);
		\draw (r) edge[-] (a2);
		\draw (b) edge[-] (e);
	\end{tikzpicture}
}
\subfigure[]{
	\begin{tikzpicture}[yscale=0.85,xscale=0.5]
		\node at (0,0) (r) {$r$};
		\node at (-0.5,-1) (a) {$a$};
		\node at (0.5,-1) (b) {$b$};
		\node at (0.5,-2) (e) {$e$};
		\draw (r) edge[-] (b);
		\draw (r) edge[-] (a);
		\draw (b) edge[-] (e);
	\end{tikzpicture}
}
\subfigure[]{
	\begin{tikzpicture}[yscale=0.85,xscale=0.5]
		\node at (0,0) (r) {$r$};
		\node at (-0.5,-1) (a) {$a$};
		\node at (0.5,-1) (b) {$c$};
		\node at (0.5,-2) (e) {$e$};
		\draw (r) edge[-] (b);
		\draw (r) edge[-] (a);
		\draw (b) edge[-] (e);
	\end{tikzpicture}
}
\subfigure[]{
	\begin{tikzpicture}[yscale=0.85]
		\node at (0,0) (r) {$r$};
		\node at (0,-1) (b) {$b$};
		\node at (0,-2) (e) {$e$};
		\draw (r) edge[-] (b);
		\draw (b) edge[-] (e);
	\end{tikzpicture}
}
\subfigure[]{
	\begin{tikzpicture}[yscale=0.85]
		\node at (0,0) (r) {$r$};
		\node at (0,-1) (b) {$c$};
		\node at (0,-2) (e) {$e$};
		\draw (r) edge[-] (b);
		\draw (b) edge[-] (e);
	\end{tikzpicture}
}
\subfigure[]{
	\begin{tikzpicture}[yscale=0.85,xscale=0.5]
		\node at (0,0) (r) {$r$};
		\node at (-0.5,-1) (a) {$a$};
		\node at (-0.5,-2) (d) {$d$};
		\node at (0.5,-1) (b) {$b$};
		\node at (0.5,-2) (e) {$e$};
		\draw (r) edge[-] (b);
		\draw (r) edge[-] (a);
		\draw (b) edge[-] (e);
		\draw (a) edge[-] (d);
	\end{tikzpicture}
}
\subfigure[]{
	\begin{tikzpicture}[yscale=0.85,xscale=0.5]
		\node at (0,0) (r) {$r$};
		\node at (0,-1) (b) {$b$};
		\node at (0.5,-2) (e) {$e$};
		\node at (-0.5,-2) (e2) {$e$};
		\draw (r) edge[-] (b);
		\draw (b) edge[-] (e);
		\draw (b) edge[-] (e2);
	\end{tikzpicture}
}
\subfigure[]{
	\begin{tikzpicture}[yscale=0.85,xscale=0.5]
		\node at (0,0) (r) {$r$};
		\node at (0,-1) (b) {$c$};
		\node at (0.5,-2) (e) {$e$};
		\node at (-0.5,-2) (e2) {$e$};
		\draw (r) edge[-] (b);
		\draw (b) edge[-] (e);
		\draw (b) edge[-] (e2);
	\end{tikzpicture}
}
\caption{\label{fig:cs}Characteristic sample for the schema $S$ from Example~\ref{ex:min:tree}.}
\end{figure}

\noindent We have proposed Algorithm~\ref{alg2}, which is a sound and complete algorithm for learning disjunctive multiplicity schemas from trees positive examples.
Thus, we can state the main result of this section:
\begin{theorem}\label{th:dms:pos}
The concept class $\dms$ is learnable in polynomial time and data from positive examples i.e., in the setting $(\Tree,\dms, L)$.
\end{theorem}

\section{Learning MS from positive examples}\label{sec:ms}
In this section we show that the MS are learnable from positive examples i.e., in the setting $(\Tree,\ms,L)$.
Recall that the MS allow no disjunction in the rules, in other words they use expressions of the form $a_1^{M_1}\shuffle\ldots\shuffle a_n^{M_n}$.
Due to this very particular form, we can \emph{capture} a MS $S=(\root_S,R_S)$ using a function $\mu:\Sigma\times\Sigma\rightarrow\{0,1,?,+,*\}$ obtained directly from the rules of $S$:
\[
a\rightarrow a_1^{\mu(a,a_1)}\shuffle\ldots\shuffle a_n^{\mu(a,a_n)}.
\]
For example, given the schema $S$ having the root $r$ and the rules:
\begin{gather*}
r\rightarrow a^+\shuffle b, \qquad a\rightarrow b^*, \qquad b\rightarrow a^?\shuffle b^?,
\end{gather*}
we have :
\begin{gather*}
\mu(r,a) = +, \qquad \mu(r,b) = 1, \qquad   ~\mu(r,r) = 0,\\
\mu(a,a) = 0, \qquad ~\mu(a,b) = {*}, \qquad \mu(a,r) = 0,\\
\mu(b,a) = {?},~ \qquad \mu(b,b) = {?}, \qquad ~\mu(b,r) = 0.
\end{gather*}
Note that given the function $\mu(\cdot)$ we can easily construct the initial $S$.
We use this characterization in Algorithm~\ref{alg3}, a polynomial and sound algorithm which learns a minimal MS from a set of trees.
We assume w.l.o.g. that all the trees from the sample have as root label the same label $r$.
If this assumption is not satisfied, the sample is not consistent.
The minimality of the algorithm follows from the minimality of the inferred multiplicity for each pair of labels $(a,b)$, using the function $\minfit(\cdot)$ (cf.\ Section~\ref{sec:dms:pos}).
Moreover, Algorithm~\ref{alg3} is complete.
We can easily construct a characteristic sample of cardinality polynomial in the size of the alphabet by using the same steps provided in the previous section, for unordered words and for trees.
\begin{algorithm}
\caption{\label{alg3}Learning MS from positive examples.}
\ALGORITHM $\learner_\ms^+(D)$\\
\INPUT A set of trees $D = \{t_1,\ldots,t_n\}$ s.t.\ $\lab_{t_i}(\root_{t_i}) = r$ (with $1\leq i\leq n)$\\
\OUTPUT A minimal MS $S$ consistent with $D$\\
\LN \FOR $a\in\Sigma$ \DO\\
\LN \TAB \LET $D'=\{\ch_t^n\mid t\in D.\ n\in N_t.\ \lab_t(n) = a\}$\\
\LN \TAB \FOR $b\in\Sigma$ \DO\\
\LN \TAB \TAB \LET $\mu(a,b) = \minfit(D',b)$\\
\LN \RETURN $S$ having the root label $r$ and captured by $\mu$
\end{algorithm}

\noindent We have proposed a sound and complete algorithm which learns a minimal MS consistent with a set of positive examples, so we can state the following result:
\begin{theorem}\label{th:ms:pos}
The concept class $\ms$ is learnable in polynomial time and data from positive examples i.e., in the setting $(\Tree,\ms, L)$.
\end{theorem}

\section{Impact of negative examples}\label{sec:dms:pos:neg}
In the previous sections, we have considered the settings where the user provides positive examples only.
In this section, we allow the user to additionally specify negative examples.
The main results of this section are that the MS are learnable in polynomial time and data in the presence of both positive and negative examples, while the DMS are not.
We use two symbols $+$ and $-$ to mark whether an example is positive or negative, and we define:
\begin{itemize}
\item $\words^\pm = \words\times\{+,-\}$,
\item $ L^\pm(E) = \{(w,+)~|~w\in L(E)\}\cup\{(w,-)\mid w\in\words~\backslash~L(E)\}$, where $E$ is a disjunctive multiplicity expression,
\item $\Tree^\pm = \Tree\times\{+,-\}$,
\item $ L^\pm(S) = \{(t,+)~|~t\in L(S)\}\cup\{(t,-)\mid t\in\Tree~\backslash~ L(S)\}$, where $S$ is a disjunctive multiplicity schema.
\end{itemize}
Formally, the setting for learning disjunctive multiplicity expressions from positive and negative examples is $(\words^\pm, \DME,  L^\pm)$, while for learning DMS from positive and negative examples we have $(\Tree^\pm, \dms, L^\pm)$.
We obtain analogously the settings for disjunction-free multiplicity expressions and schemas: $(\words^\pm,\ME,L^\pm)$ and $(\Tree^\pm, \ms,L^\pm)$, respectively.

We study the problem of checking whether there exists a concept consistent with the input sample because any sound learning algorithm needs to return \emph{null} if and only if there is no such concept.
Therefore, consistency checking is an easier problem than learning and its intractability precludes learnability.
Formally, given a learning setting $\mathcal K = (\mathcal E,\mathcal C,\mathcal L)$, the $\mathcal K$\emph{-consistency} is the following decision problem:
\[
\CONS_{\mathcal K} = \{D\subseteq \mathcal E\mid\exists c\in\mathcal C.\ D\subseteq\mathcal L(c)\}.
\]
Note that the consistency checking is trivial when only positive examples are allowed.
For instance, if we want to learn disjunctive multiplicity expressions from positive examples over the alphabet $\{a_1,\ldots,a_n\}$, the disjunctive multiplicity expression $a_1^*\shuffle\ldots\shuffle a_n^*$ is always consistent with the examples.
When we also allow negative examples, the problem becomes more complex, particularly in the case of disjunctive multiplicity expressions and schemas, where this problem is not tractable.

First, we show that the consistency checking is tractable for MS.
In Section~\ref{sec:ms}, we have proposed Algorithm~\ref{alg3}, which learns a minimal MS consistent with a set of positive examples.
Note that, given a set of trees, there exists a \emph{unique minimal MS} consistent with them.
The argument is that Algorithm~\ref{alg3} uses the function $\minfit(\cdot)$ (cf.\ Section~\ref{sec:dms:pos}) to infer minimal multiplicities which are unique and sufficient to capture a MS.
Thus, the consistency checking becomes trivial for MS: given a sample containing positive and negative examples, there exists a MS consistent with them iff no tree used as negative example satisfies the minimal MS returned by Algorithm~\ref{alg3}.
Consequently, we easily adapt Algorithm~\ref{alg3} to handle both positive and negative examples and we propose Algorithm~\ref{alg4}.
\begin{algorithm}
\caption{\label{alg4}Learning MS from positive and negative examples.}
\ALGORITHM $\learner_\ms^\pm(D)$\\
\INPUT A sample $D = \{(t,\alpha)\mid t\in\Tree, \alpha\in\{+,-\}\}$\\
\OUTPUT A minimal MS $S$ such that $D\subseteq  L^\pm(S)$, or \emph{null} if no such schema exists\\
\LN \LET $D' = \{t\in\Tree\mid (t,+)\in D\}$\\
\LN \LET $S=\learner_\ms^+(D')$\\
\LN \IF $\exists t\in\Tree.\ (t,-)\in D\wedge t\in L(S)$ \THEN \\
\LN \TAB \RETURN \emph{null}\\
\LN \RETURN $S$
\end{algorithm}

\noindent Essentially, Algorithm~\ref{alg4} returns the minimal schema consistent with the positive examples iff there is no negative example satisfying it, and otherwise it returns \emph{null}.
Note that Algorithm~\ref{alg4} is sound and works in polynomial time in the size of the input.
The completeness of Algorithm~\ref{alg4} follows from the completeness of Algorithm~\ref{alg3}.
Given a MS $S$, we can construct a characteristic sample $\cs_S$ that contains only positive examples, analogously to how it is done for Algorithm~\ref{alg3}.
We have proposed a polynomial, sound, and complete  algorithm which learns minimal MS from positive and negative examples, so we state the first result of this section:
\begin{theorem}\label{th:ms:posneg}
The concept class $\ms$ is learnable in polynomial time and data from positive and negative examples i.e., in the setting $(\Tree^\pm,\ms, L^\pm)$.
\end{theorem}
Next, we prove that the concept class $\dms$ is not learnable in polynomial time and data in the setting $\dms^\pm=(\Tree^\pm, \dms, L^\pm)$.
For this purpose, we first show the intractability of learning disjunctive multiplicity expressions from positive and negative examples i.e., in the setting $\DME^\pm = (\words^\pm,\DME,L^\pm)$.
We study the complexity of checking the consistency 
of a set of positive and negative examples and we prove the intractability of $\CONS_{\DME^\pm}$.
Intuitively, this follows from the fact that, given a set of unordered words, there may exist an exponential number of minimal consistent disjunctive multiplicity expressions, and we may need to check all of them to decide whether there exist negative examples satisfying them.
Formally, we have the following result:
\begin{lemma}\label{lemma:cons}
$\CONS_{\DME^\pm}$ is NP-complete.
\end{lemma}
\begin{proof}\normalfont
We prove the NP-hardness by reduction from $\SAT$ which is known as being NP-complete.
We take a formula $\varphi$ in 3CNF containing the clauses $c_1,\ldots,c_k$ over the variables $x_1,\ldots,x_n$.
We generate a sample $D_\varphi$ over the alphabet $\Sigma=\{t_1, f_1,\ldots,t_n,f_n\}$ such that:
\begin{itemize}
\item $(t_1f_1\ldots t_nf_n, +) \in D_\varphi$,
\item $(\varepsilon, -) \in D_\varphi$,
\item $(t_if_i, +), (t_it_if_if_i, -)\in D_\varphi$, for $1\leq i\leq n$,
\item $(w_j, -)\in D_\varphi$, where $w_j= v_{j1}v_{j1}v_{j2}v_{j2}v_{j3}v_{j3}$, for any $j$ such that $1\leq j\leq k$, where $x_{j1}, x_{j2}, x_{j3}$ are the literals used in the clause $c_j$ and for any $l$ such that $1\leq l\leq 3$, $v_{jl}$ is $t_{jl}$ if $x_{jl}$ is a negative literal in $c_j$, and $f_{jl}$ otherwise.
\end{itemize}
For example, for the formula $(x_1\vee\neg x_2\vee x_3)\wedge(\neg x_1\vee x_3\vee \neg x_4)$, we generate the sample:
\[
\begin{tabular}{cc}
$(t_1f_1t_2f_2t_3f_3t_4f_4, +)$, & $(\varepsilon, -)$,\\
$(t_1f_1,+)$,& $(t_1t_1f_1f_1, -)$,\\
$(t_2f_2,+)$,& $(t_2t_2f_2f_2, -)$,\\
$(t_3f_3,+)$,&$(t_3t_3f_3f_3, -)$,\\
$(t_4f_4,+)$,&$(t_4t_4f_4f_4, -)$,\\
& $(f_1f_1t_2t_2f_3f_3, -)$,\\
& $(t_1t_1f_3f_3t_4t_4, -)$.
\end{tabular}
\]
For a given $\varphi$, a valuation is a function $V:\{x_1,\ldots,x_n\}\rightarrow \{\true,\false\}$.
Each of the $2^n$ possible valuations encodes a minimal disjunctive multiplicity expression $E_V$ consistent with the positive examples from $D_\varphi$, constructed as follows:
\[
E_V = (v_1\mid\ldots\mid v_n)^+\shuffle \overline{v_1}^?\shuffle \ldots \shuffle\overline{v_n}^?,
\]
where, for $1\leq i\leq n$, if $V(x_i)=\true$ then $v_i=t_i$ and $\overline{v_i}=f_i$.
Otherwise, $v_i=f_i$ and $\overline{v_i}=t_i$.
Next, we show that, for any valuation $V$, $V\models\varphi$ iff $E_V$ is consistent with $D_\varphi$.

For the \emph{only if} case, consider a valuation $V$ such that $V\models\varphi$ and we take the corresponding expression $E_V = (v_1\mid\ldots\mid v_n)^+\shuffle \overline{v_1}^?\shuffle \ldots \shuffle\overline{v_n}^?$.
Note that $t_1f_1\ldots t_nf_n$ and all $t_if_i$'s (with $1\leq i\leq n$) satisfy $E_V$, while $\varepsilon$ does not satisfy $E_V$.
Also note that for $1\leq i\leq n$, one symbol between $t_i$ and $f_i$ occurs at least once, while the other occurs at most once, so all $t_it_if_if_i$'s do not satisfy $E_V$.
Assume that there is a $w_j$ (with $1\leq j\leq k$) such that $w_j$ satisfies $E_V$, which by construction implies that the clause $c_j$ is not satisfied by the valuation $V$, which implies a contradiction.
Hence, $w_j$ does not satisfy $E_V$ for any $1\leq j\leq k$.
Therefore, $E_V$ is consistent with $D_\varphi$.

For the \emph{if} case, we assume that $E_V$ is consistent with the sample $D_\varphi$.
Since the $w_j$'s (with $1\leq j\leq k$) encode the valuations making the clauses $c_j$'s false and none of the $w_j$'s satisfies $E_V$, then the valuation $V$ encoded in $E_V$ makes the formula $\varphi$ satisfiable.

The construction of $D_\varphi$ also ensures that if there exists a disjunctive multiplicity expression consistent with $D_\varphi$, it has the form of $E_V$.
Therefore, $\varphi\in\SAT$ iff $D_\varphi \in \CONS_{\DME^\pm}$.

To prove the membership of $\CONS_{\DME^\pm}$ to NP, we point out that a Turing machine guesses a  disjunctive multiplicity expression $E$, whose size is linear in $|\Sigma|$ since repetitions are discarded among the disjunctions of $E$.
Moreover, checking whether $E$ is consistent with the sample can be easily done in polynomial time.
\qed\end{proof}
We extend the above result to $\CONS_{\dms^\pm}$:
\begin{corollary}
$\CONS_{\dms^\pm}$ is NP-complete.
\end{corollary}
\begin{proof}\normalfont
The NP-hardness of $\CONS_{\DME^\pm}$ implies the NP-hardness of $\CONS_{\dms^\pm}$:
it is sufficient to consider flat trees having all the same root label. 
Moreover, to prove the membership of $\CONS_{\dms^\pm}$ to NP, a Turing machine guesses a disjunctive multiplicity schema $S$, whose size is polynomial in $|\Sigma|$, and checks whether $S$ is consistent with the sample (which can be done in polynomial time).
\qed\end{proof}
Since consistency checking in the presence of positive and negative examples is intractable for DMS, we conclude that:
\begin{theorem}\label{th:dms:pos-neg}
Unless P = NP, the concept class $\dms$ is not learnable in polynomial time and data from positive and negative examples i.e., in the setting $(\Tree^\pm, \dms,  L^\pm)$.
\end{theorem}

\section{Conclusions and future work}\label{sec:conclusions}
We have studied the problem of learning unordered XML schemas from examples given by the user. 
We have investigated the learnability of DMS and MS in two settings: one allowing positive examples only, and one that allows both positive and negative examples.
To the best of our knowledge, no research has been done on learning unordered XML schema formalisms, nor on allowing both positive and negative examples in the process of schema learning.
We have proven that the DMS are learnable only from positive examples, and we have shown that they are not learnable from positive and negative examples by using the intractability of the consistency checking.
Moreover, we have proven that the MS are learnable in both settings: from only positive examples, and also from positive and negative examples.
For all the learnable cases we have proposed learning algorithms that return minimal schemas consistent with the examples.

As future work, we want to use a more specific learnability condition i.e., to require the size (instead of the cardinality) of the characteristic sample to be polynomial in the size of the alphabet.
Thus, we will fully adhere to the classical definition of the characteristic sample in the context of grammatical inference~\cite{Higuera97}.
Our preliminary research indicates that we are able to do this by using a compressed representation of the XML documents with directed acyclic graphs~\cite{LoMaNo13}.
The learning algorithms that we propose in this paper will work without any alteration.
Moreover, we would like to extend our learning algorithms for more expressive unordered schemas, for instance schemas which allow \emph{numeric occurrences}~\cite{KiTu07} of the form $a^{[n,m]}$ that generalize multiplicities by requiring the presence of at least $n$ and at most $m$ elements $a$.
Additionally, we want to use the learning algorithms for unordered schemas to boost the existing learning algorithms for twig queries~\cite{StWi12}.
For this purpose, we have to investigate first the problem of query minimization~\cite{ACLS02} in the presence of DMS.
Next, we want to propose a twig query learning algorithm which infers the schema of the documents and then it uses the schema to improve the quality of the learned twig query.

\end{document}